\documentclass[runningheads]{llncs}

\usepackage[T1]{fontenc}
\usepackage{xcolor}
\usepackage{amsmath}

\usepackage{tikz-cd}
\usepackage{mathrsfs}
\usetikzlibrary{arrows.meta,calc}

\usepackage{graphicx}

\usepackage{amssymb}
\begin{document}

\title{Geometry of Cells Sensible to Curvature and Their Receptive Profiles}

\author{Vasiliki Liontou \inst{1}\orcidID{0000-0002-1851-4101} }

\authorrunning{V. Liontou}

\institute{University of Bologna, Bologna, Italy, 
\email{vasiliki.liontou@unibo.it}}

\maketitle              

\begin{abstract}

We propose a model of the functional architecture of curvature sensible cells in the visual cortex that associates curvature with scale.  The feature space of orientation and position is naturally enhanced via its oriented prolongation, yielding a 4-dimensional manifold endowed with a canonical Engel structure. This structure encodes position, orientation, signed curvature, and scale. We associate an open submanifold of the prolongation with the quasi-regular representation of the similitude group $SIM(2)$, and find left-invariant generators for the Engel structure. Finally, we use the generators of the Engel structure to characterize curvature-sensitive receptive profiles .

\end{abstract}

\section{Introduction}
The mammalian visual cortex is organized in 
families of cells, each one sensitive to a specific feature of the stimulus image: position, orientation,
scale, color, curvature, movement, stereo and many others. The selectivity of a cell to a differential feature introduces a differential
constraint in the space. The prototype example is the contact 3-manifold underlying the model of $V_1$  in \cite{Tondut},\cite{FuncArch},\cite{Pet}, which is of the form $M = \mathbb{S}(T^*\mathbb{E}^2)$, namely the unit sphere bundle of the cotangent bundle of the Euclidean plane $\mathbb{E}^2$, also known as the manifold of contact elements of $\mathbb{E}^2$. Under this assumption, the topological feature space of all possible positions and orientations is identified with $\mathbb{R}^2\times \mathbb{S}^1$. Moreover, to obtain the whole family of orientation-position receptive profiles from the mother window one should apply the group action of the special Euclidean group $SE(2)$ on $\mathbb{R}^2$.The group $SE(2)$ describes the orientation preserving isometries of $\mathbb{E}^2$, which are the translations and rotations.
The Lie group embedding of $SE(2)$ to the group of diffeomorphisms of $\mathbb{R}^2$
is the group action of $SE(2)$ on $\mathbb{R}^2$. Via this group action the vector fields generating the contact distribution of $M$ can be pushed-forward to vector fields of $\mathbb{R}^2$, introducing non-commutative differential operators on receptive profiles.  Extending the idea that receptive profiles are minimizers of uncertainty principle, which was firstly introduced for orientation-spatial frequency and position in \cite{Daug} , the authors of \cite{Uncertainty} prove that the receptive profiles for orientation and position are the minimizers
of the uncertainty principle with respect to these differential operators. 

The structure of orientation-position sensible cells in $V_1$ indicates the existence of a framework applicable to other groups of cells which also describes the modular structure of the visual cortex, as described in \cite{CortArch}. The modular structure is compatible
with the experiments Hubel and Wiesel in  \cite{Hubel}. This framework includes a differential constraint governing the connectivity between cells and a Lie group whose symmetries are learned from the symmetries of the stimulus. Moreover, this Lie group attains an action of the retinal surface. If the differential constraints are consistent with the Lie group structure- for instance, the contact 1-form of orientation-position cells is left-invariant with respect to $SE(2)$- then we can use the Lie group action to find non-commuting vector fields on $\mathbb{R}^2$. These vector fields impose an uncertainty principle, and the receptive profiles are the minimizers.

On the other hand, the authors in \cite{Scale} use the symplectization of $\mathbb{S}(T^*\mathbb{E}^2)$ to take into account the scale, obtaining differential constraints for the enlarged space induced by the preceding contact structure (a method first applied in \cite{Tondut} and extended in \cite{Scale}, \cite{LioMarc}). In \cite{Duits} the Lie group of similitudes was associated to the enlarged space of scale.  The authors use the similitude group $SIM(2)$, which appeared as the Lie group of the Lie algebra of differential constraints in \cite{Scale}, to construct a group-wavelet transform. 
Then they perform left-invariant diffusion in the range of the 
$SIM(2)$ transform and show that by changing the coefficient of the metric on the scale axis, the curvature of the level sets of the Heat Kernel are affected. This result indicates a relation between curvature and scale.  
 
Models of curvature have employed Engel or Engel-type structures on the space $\mathbb{R}^2 \times SO(2) \times \mathbb{R}$ (\cite{NeuroBook}, \cite{Liouville}, \cite{Pet}) to introduce differential constraints. However, unlike in other cell families, no associated Lie group structure has yet been identified. The core difficulty lies in the fact that, unlike the horizontal vector fields on $M$, the horizontal vector fields of the Engel structure generate an infinite-dimensional Lie algebra.

In this work, we provide a model of cells sensible to curvature, 
and their receptive profiles as a circle bundle over the manifold of orientation-position sensible cells, that associates curvature with scale.  
The retina is modeled as a surface endowed with a Riemannian metric tensor, accounting for the log-polar mapping. The circle bundle is the prolongation of the contact manifold orientation-position sensible cells and the differential constraints are given by the canonical Engel structure of the prolongation.
Previous models for scale are contained as open submanifolds of this model, away from curvature 0. As a result, on this submanifold the Engel structure is left-invariant with respect to the similitude group $SIM(2)$. 
Moreover, we use the action of $SIM(2)$ on $\mathbb{R}^2$ to obtain differential operators that characterize the receptive profiles.
The scope is to provide a model for curvature that reflects the modular structure of the visual cortex, where the structure of each group of cells is obtained from the previous one, as described in \cite{CortArch}.

\section{Preliminaries} 

We consider the retina $\mathcal{S}$ as a Riemannian 2-dimensional manifold, diffeomorphic to $\mathbb{R}^2$, equipped with a metric tensor $g$, see \cite{LogPolar}. 
 The orientation–position feature space is naturally identified with the manifold of \textbf{oriented contact elements }of $\mathcal{S}$, given by 
\begin{equation*}
M = T^*\mathcal{S} \setminus \{0\} / \mathbb{R}^+ \simeq \mathbb{R}^2 \times \mathbb{S}^1
\end{equation*}
where $\{0\}$ denotes the zero section of the cotangent bundle $T^*\mathcal{S}$ and $\mathbb{R}^+$ acts by scalar multiplication on covectors \cite{Tondut}, \cite{FuncArch}, \cite{LioMarc}. 

Since $\mathcal{S}$ is endowed with a metric tensor $g$, the space of oriented contact elements is canonically identified with the \textbf{unit cosphere bundle} $\mathbb{S}(T^*\mathcal{S})$. Via the canonical isomorphism induced by the metric $g$, each unit covector $q\in \mathbb{S}(T^*\mathcal{S})$ corresponds to a unit vector $X\in \mathbb{S}\mathcal{S}$ satisfying $$q(-)=g(X,-).$$ Thus, at each point $p=(x,y)\in \mathcal{S}\simeq \mathbb{R}^2$, we introduce local coordinates $(x,y,\theta)$ on $\mathbb{S}(T^*\mathbb{R}^2)$, where the angle function  $\theta_{(x,y)}: \mathbb{S}(T_{(x,y)}^*\mathcal{S})\rightarrow (0,2\pi)$ measures the angle between the unit vector field $X$ (associated to $q$) and the coordinate vector field $\partial_x$ at $(x,y)$, namely 
\begin{equation}
    \theta(q)= Arg(e^{i g(X, \partial_x)}), \text{ for } X\in \mathbb{S}(T_{(x,y)}\mathbb{R}^2) \text{ with } X\neq \partial_x.
\end{equation}
 The horizontal connectivity is determined by contact distribution $$\tau= span \{\partial_\theta, cos(\theta)\partial_x+ sin(\theta)\partial_y\}$$ which is locally the kernel of the contact 1-form expressed in coordinates $(x,y,\theta)$
\begin{equation}
    a= -sin(\theta)dx+cos(\theta)dy.
\end{equation}

Given a distribution $\mathcal{D}$ of $k-planes$ on a manifold $M$, one can consider the \textit{oriented projectivization} $\mathbb{S}M$ of $\mathcal{D}$ which is the fiber bundle over $M$ with typical fiber $\mathbb{S}_qM$ the half-lines of $\mathcal{D}_q$.  In this case, the distribution of interest is the contact distribution $\tau$. The oriented projectivization of $M$ is diffeomorphic to $M\times \mathbb{S}^1$.  The projectivization $\mathbb{S}M$ inherits a canonical 2-dimensional plane distribution defined by declaring that a curve $(q(t),\ell(t))\in M\times\mathbb{S}^1$ is horizontal if and only if the derivative of the point of contact $\dot{q}(t)$ lies on the line $\ell(t)$ for every $t$. Equivalently, if $\pi: M\times \mathbb{S}^1\rightarrow M$ is the projection of the fiber bundle and $d\pi: T( M\times \mathbb{S}^1)\rightarrow T( M) $ its differential, the distribution $\mathcal{D}$ is described as follows
\begin{equation}
    \mathcal{D}=\{d\pi^{-1}_{(q,\ell)}(\ell_q): (q,\ell) \in  M\times \mathbb{S}^1 \}.
\end{equation}
The projectivization $\mathbb{P}M$ (or the oriented projectivization $\mathbb{S}M$) together with the distribution $\mathcal
{D}$ is the \textbf{ (oriented) prolongation} of $M$ \cite{Montgomery}.The vector fields $X_1=cos(\theta)\partial_x+sin(\theta)\partial_y$ and $X_2=\partial_\theta$  form linear coordinates on each contact plane $\tau_q$ and therefore a ray $\ell_q \subset \tau_q$ can be expressed with respect to these vector fields as
\begin{equation}
    \ell_q= r(s) X_1(q)+ \kappa(s) X_2(q), ~ s>0.
\end{equation}
A curve $\Gamma(t)=\{(x(t),y(t),\theta(t),\ell(t)), ~t\in \mathbb{R}\}$ on $M\times\mathbb{S}^1$ is horizontal if and only if the tangent vector $\dot{x}(t)\partial_x+ \dot{y}(t)\partial_y+\dot{\theta}(t)\partial_\theta$ on the projection $\pi(\Gamma(t))\in M$ lies on the ray $\ell(t)=r(s,t)X_1(q(t))+ \kappa(s,t)X_2(q(t))$. Thus the horizontality condition reads as
\begin{align*}
    \dot{x}(t)= r(s,t)cos(\theta(t)),~
    \dot{y}(t)=r(s,t)sin(\theta(t)),\\
    \dot{\theta}(t)=\kappa(s,t),~
    \ell(t)\in \mathbb{S}^1.
\end{align*}
Using affine coordinates on $\tau$ with respect to the frame $\{X_1,X_2\}$, the distribution $\mathcal{D}$ is spanned by the vector fields
\begin{equation*}
   \mathcal{X}= rcos(\theta)\partial_x+rsin(\theta)\partial_y+ \partial_\theta, ~ \mathcal{R}=\partial_r
\end{equation*}
on the open submanifold $\mathbb{S}M\setminus \{\kappa\neq 0\}$ with coordinates $(x,y,\theta,r)$ and by 
\begin{equation*}
     \mathcal{X}= cos(\theta)\partial_x+sin(\theta)\partial_y+ \kappa\partial_\theta, ~ \mathcal{K}=\partial_\kappa
\end{equation*}
on the open submanifold $\mathbb{S}M\setminus \{r\neq 0\}$ with coordinates $(x,y,\theta,\kappa)$.
\section{Curvature Feature Space} 
\subsection{Engel Structure on the Curvature Space}
We use the circle bundle $M\times \mathbb{S}^1$ with co-rank 2 distribution $\mathcal{D}$, which is the oriented projectivization of the orientation-position contact manifold $(M=\mathbb{S}(T^*\mathcal{S}),\tau)$ to describe the singed curvature space.  
\begin{theorem}\label{Theorem}
    The feature space of position-orientation and signed curvature is the oriented prolongation $\mathbb{S}M$ of the feature space of orientation-position $(M,\tau)$. 
\end{theorem}
\begin{proof}(Sketch)
Let $\gamma(t)=(x(t),y(t))$ be a curve on the retina $\mathcal{R}=(\mathbb{R}^2, g_{\mathbb{E}})$, expressed in standard coordinates. The lift $\tilde{\gamma}$ of $\gamma$ to the orientation-position feature space defined as $\tilde{\gamma}(t)=\{q(t)\in \mathbb{S}(T^*\mathbb{R}^2): q(t)(\dot{\gamma}(t))=0\}$ is horizontal with respect to the contact distribution $\tau$. In local coordinates $(x,y,\theta)$, the lift  is $\tilde{\gamma}=q(t)=(x(t),y(t),\theta(t))$ where $(x(t),y(t))$ is the planar position while $\theta(t)$ is the angle of the tangent vector $\dot{\gamma}(t)$ with $\partial_x$. Now, by definition the signed curvature $\mathfrak{k}_\gamma$ of $\gamma$ is the rate of change of $\theta(t)$, $\mathfrak{k}_\gamma(t)=\dot{\theta}(t)=\kappa(t,s)$.  Therefore, the horizontal lift $\Gamma(t)=(q(t),\ell(t))$ of $\tilde{\gamma}$ on the prolongation $\mathbb{S}M$ is \begin{equation*}
    \Gamma(t)=(x(t),y(t),\theta(t),\dot{\theta}(t))=(x(t),y(t),\theta(t),\mathfrak{k}_\gamma(t)).
\end{equation*}
\end{proof}
The rank 2-distribution $\mathcal{D}$ is an Engel structure, \cite{Montgomery}. Engel structures form another class of nonintegrable distributions which
is closely related to contact structures. By definition, an Engel structure is a
smooth distribution $\mathcal{D}$ of rank 2 on a manifold  of dimension 4 which satisfies
the nonintegrability conditions
\begin{equation}
   rank [\mathcal{D},\mathcal{D}]=3~,~ rank[\mathcal{D},[\mathcal{D},\mathcal{D}]]=4,
\end{equation}
 where $[\mathcal{D}, \mathcal{D}]$ consists of those tangent vectors which can be obtained by taking
commutators of local sections of $\mathcal{D}$. 
In fact, prolongations and oriented prolongations of contact manifolds such as $\mathbb{S}\tau$ are the canonical examples of Engel structures. More specifically, for the absolute curvature-position-orientation space $(\mathbb{P}\tau, \mathcal{D})$ we have on the the local chart $\{r \neq 0\}$
\begin{align*}
 [\mathcal{D},\mathcal{D}
 ]= span\{\mathcal{X},\mathcal{K}, \mathcal{Y}=[\mathcal{X},\mathcal{K}]=\partial_\theta\}\\
 [\mathcal{D},[\mathcal{D},\mathcal{D}]]=span\{\mathcal{X},\mathcal{K},\mathcal{Y}, \mathcal{Z}=[\mathcal{X},\mathcal{Y}]=-sin(\theta)\partial_x+cos(\theta)\partial_y\}.
\end{align*}
In the table bellow we summarize the mechanism of consecutive prolongations from the retina that leads to the circle bundle $\mathbb{S}_\tau$ of Theorem \ref{Theorem} which model the oriented curvature and scale feature space.
\[
\begin{tikzcd}[row sep=small]
\begin{array}{c}
\mathbb{S}\tau = (M \times \mathbb{S}^1 ,\mathcal{D})\\
\scriptstyle\text{(signed curvature-position-orientation)}
\end{array}
\arrow[d] \\
\begin{array}{c}
\mathbb{P}\tau = M \times \mathbb{P}^1 \\
\scriptstyle\text{(absolute curvature-position-orientation)}
\end{array}
\arrow[d] \\
\begin{array}{c}
M=(\mathcal{S} \times \mathbb{S}^1,\tau) \\
\scriptstyle\text{(position-orientation)}
\end{array}
\arrow[d] \\
\begin{array}{c}
\mathcal{S} \\
\scriptstyle\text{(position)}
\end{array}
\end{tikzcd}
\quad
\begin{tikzcd}[row sep=large]
\begin{array}{c}
\scriptstyle\text{(Oriented Prolongation, Engel Structure)}
\end{array} \\
\begin{array}{c}
\scriptstyle\text{(Prolongation, Engel Structure)}
\end{array} \\
\begin{array}{c}
\scriptstyle\text{(contact elements of $\mathcal{S}$, $SE(2)$-transform)}
\end{array} \\
~
\end{tikzcd}
\]
In regularity theory or more generally in local problems, the important property is that the Engel distribution $\mathcal{D}$ is bracket generating. In fact, locally all Engel structures are equivalent in a similar way that all contact structures are locally equivalent. 
\begin{lemma}
    The closure of $\mathcal{D}$ under the Lie bracket is an infinite-dimensional Lie algebra.
\end{lemma}
The goal is to associate a finite Lie group with the prolongation $\mathbb{P}\tau$ (or the oriented prolongation $\mathbb{S}\tau$) such that the Engel structure $\mathcal{D}$ is left-invariant. We introduce, a new pair of vector fields in the open submanifold $U\subset\mathbb{P}\tau$ which is the intersection of the affine charts $$U=\{(q,\ell)\in \mathbb{S}\tau: \ell_q=rX_1(q)+\kappa X_2(q) \text{ and } \kappa \cdot r \neq 0\}\simeq \mathbb{R}^2\times \mathbb{S}^1\times \mathbb{R}^+.$$ In this open subset of $\mathbb{S}\tau$, we can consider a new pair of generators for the distribution $\mathcal{D}$, namely in local coordinates $(x,y,\theta, r)$ the new generators are
\begin{equation}\label{vectorfields}
    \mathcal{X}_{loc}= \mathcal{X}= r cos(\theta)\partial_x+rsin(\theta)\partial_y+\partial_\theta, ~ \mathcal{R}_{loc}=r\mathcal{R}=r \partial_r.  
\end{equation}

\subsection{Curvature, Scale and the Similitude Group}
In models of the cells sensible to scale (\cite{Scale}, \cite{Scale2}, \cite{Duits}) the underlying manifold is $\mathbb{R}^2\times\mathbb{S}^1\times\mathbb{R}^+$ and the differential constraints are given by a differential 2-form which is derived from the contact 1-form on $SE(2)$ via symplectization.  The commutation rules associated to the vector fields that arise from this differential constraint, show that the Lie group associated to the scale is the group of similitudes SIM(2).  
 The similitude group is the semidirect product $SIM(2)=\mathbb{R}^2\rtimes (SO(2)\times \mathbb{R}^+)$, hence $SIM(2)$ is the set $\mathbb{R}^2\times \mathbb{S}^1\times \mathbb{R}^+$ with group law
\begin{equation*}
    ((x,y),\theta,\delta)\cdot ((x^\prime,y^\prime),\theta^\prime, \delta^\prime)=((x,y)+\delta\theta (x^\prime,y^\prime), \theta+\theta^\prime, \delta \delta^\prime).
\end{equation*}
The group can be identified with the group generated by planar translations $T_{(x,y)}$, planar rotations $R_\theta$ and dilations $D_\delta$ and therefore it acts transitively on $\mathbb{R}^2$ giving rise to the quisi-regular representation on $L^2(\mathbb{R}^2)$ via 
\begin{align*}
   \pi(x,y,\theta,\delta)f(z)=D_\delta R_\theta T_{(x,y)} f(z)= |\delta|^{-1}f(\delta^{-1}R_\theta(z-(x,y))), z\in\mathbb{R}^2.
\end{align*}
Indeed, let $(x_0,y_0)$ be a point on a planar curve $\gamma(t)=(x(t),y(t))$, let $\theta$ be the angle of $\dot{\gamma}(t)$ at $(x_0,y_0)$ with $\partial_x$ and let $C$ be the osculating circle of $\gamma$ at point $(x_0,y_0)$. If $\mathfrak{k}_\gamma(t)$ is the signed curvature of $\gamma$, the radius of the osculating circle is  $R_C(t)=\frac{1}{|\mathfrak{k}_\gamma(t)|}$.

Thus, the translation $T_{-(x_0,y_0)}$ and rotation $R_{-\theta}$ translate the curve to the origin of the plane and rotate it such that $\dot{\gamma}(t)$ at $(x_0,y_0)$ forms a $0$ degree angle with $\partial_x$. Consequently, the osculating circle at $(x_0,y_0)$ is now tangent to the horizontal axis at the origin with center $(0,R)$. The dilation $D_\delta$ scales the radius of the osculating circle $R_C(t)\mapsto \delta R_C(t)$.

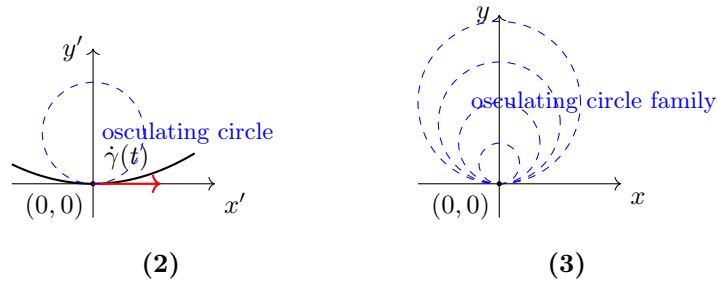
\begin{center}
\begin{tikzpicture}[scale=0.9]

\begin{scope}[shift={(2,0)}]
\node at (1,-1.2) {\textbf{(2) }};
\draw[->] (-1.2,0) -- (1.8,0) node[below right] {$x'$};
\draw[->] (0,-0.5) -- (0,2) node[left] {$y'$};

\draw[thick, domain=-1.2:1.5, smooth, variable=\x] 
  plot({\x}, {0.2*\x*\x});

\filldraw[black] (0, 0) circle (0.03);
\node[below left] at (0, 0) {$(0,0)$};

\draw[->, thick, red] (0,0) -- (1,0);
\node[above] at (0.5,0.05) {$\dot{\gamma}(t)$};

\draw[blue, dashed] (0,0.75) circle (0.75);
\node[blue] at (1.4,0.75) {\small osculating circle};
\end{scope}

\begin{scope}[shift={(8,0)}]
\node at (1,-1.2) {\textbf{(3) }};
\draw[->] (-1.2,0) -- (1.8,0) node[below right] {$x$};
\draw[->] (0,-0.5) -- (0,2.5) node[left] {$y$};

\foreach \r in {0.3, 0.6, 0.9, 1.2} {
  \draw[blue, dashed] (0,\r) circle (\r);
}

\node[blue] at (1.4, 1.2) {\small osculating circle family};

\filldraw[black] (0, 0) circle (0.03);
\node[below left] at (0, 0) {$(0,0)$};

\end{scope}

\end{tikzpicture}
\end{center}
\begin{lemma}
  Let $\gamma=(x(t),y(t))$ a planar curve and let
  $\Gamma=((x(t),y(t),\theta(t),\ell(t))$ be its horizontal curve on the prolongation $(\mathbb{S}\tau, \mathcal{D})$. If $\ell(t)=r(t)X_1+\kappa(t)X_2$ and $r(t)\cdot \kappa(t)\neq 0$ for every $t$, then $\frac{|r(t)|}{\kappa(t)}$ is the radius of the osculating circle $C_\gamma(t)$ of $\gamma$ at $\gamma(t)$.  
\end{lemma}

Now, we can consider the parameter of dilations $D_\delta$ to be the radius of the osculating circle, 
$
    \delta= \frac{|r(t)|}{\kappa(t)}
$
and the space of non-zero radii (or inverse curvatures) to be $\mathbb{R}^+$ with multiplicative group law.

\begin{proposition}
The Lie algebra $\mathfrak{g}=(\mathcal{D}^{(3)}_{|_U}, [-,-]_{|_U})$ is isomorphic to $Lie(SIM(2))$. Moreover, the Engel distribution $\mathcal{D}_U:=\{\mathcal{D}_{(q,\ell)}:(q,\ell)\in U\}$ is left-invariant with respect to $SIM(2)$. 
 
\end{proposition}
The proof follows immediately from the fact that $\mathcal{X}_{loc}$ and $\mathcal{R}_{loc}$ are linear combinations of the left-invariant vector fields corresponding tothe basis of the Lie algebra $\mathfrak{sim}(2)$, as calculated in \cite{Duits}.

\subsection{Receptive Profiles}
In the previous paragraph we established that the vector fields which span the Engel distribution away from the singularities $r=0$ and $k=0$ are left-invariant with respect to the $SIM(2)$ group. In general, the receptive profiles are obtained as minima of an uncertainty principle on $L^2(\mathbb{R}^2)$ with respect to suitable vector fields on $\mathbb{R}^2$. In \cite{Duits}, the vector fields associated to scale were chosen and the minimizers of this uncertainty principle were suitable for curvatures. Here, we use the vector fields generating the Engel structure, which we push-forward with the action of $SIM(2)$ on $\mathbb{R}^2$. These result to different differential operators and therefore to a different equation of  
receptive profiles. 
Using representation of $SIM(2)$ in $L^2(\mathbb{R}^2)$ and an appropriately chosen mother filter $\Psi^c_0: \mathbb{R}^2\rightarrow \mathbb{R}$
one can generate the family of curvature-orientation position receptive profiles 
\begin{equation*}
    \{\pi(x,y,\theta,\delta)\Psi^c_0:(x,y,\theta,\delta)\in SIM(2)\}
\end{equation*}
where $\mathbb{R}^2$ parametrizes position, $SO(2)$ parametrizes orientation and $\mathbb{R}^+$ has two possible interpretations, as the parameter space of scale (\cite{Duits}, \cite{Scale}) or as the parameter space of curvature (\cite{Liouville}, \cite{Pet}). 

At the moment, we merely require $\Psi^c_0$ that it be a Schwartz‑class function on $\mathbb{R}^2$ whose quasi‑regular transforms yield a frame of $L^2(\mathbb{R}^2)$.
The image of
the left-invariant vector fields in (\ref{vectorfields}) under the differential of the action $SIM(2)\xrightarrow{\phi} Diff(\mathbb{R}^2)$
is
\begin{equation*}
  d\phi(\mathcal{X}_{loc})= -\eta \partial_\xi+\xi\partial_\eta+\partial_\xi, ~ d\phi(\mathcal{R}_{loc})= \eta\partial_\eta +\xi\partial_\xi  
\end{equation*}
where $\xi=\frac{1}{r}sin(\theta)(x-x_0)+\frac{1}{r}cos(\theta)(y-y_0)$ and $\eta=\frac{1}{r}-sin(\theta)(x-x_0)+\frac{1}{r}cos(\theta)(y-y_0)$
and provides an algebra representation, given by the differential operators on $\mathbb{R}^2$
that correspond to the directional derivatives. 
Following the methodology of \cite{CortArch}, the shape of a curvature receptive profile is determined by the minima of the uncertainty principle with respect to $d\phi(\mathcal{X}_{loc})$ and $d\pi(\mathcal{R}_{loc})$, namely the receptive profile $\Psi_0^c$ should satisfy the equation
\begin{equation}
    d\phi(\mathcal{X}_{loc})\Psi^c_0=-id\pi(\mathcal{R}_{loc}\Psi_0^c)
\end{equation}

\section{Conclusion}   We introduce a model for curvature which combines the differential constraints of Engel structures found in previous models and the symmetry of the $SIM(2)$-group. Until now, the $SIM(2)$ was only related with scale and our model clarifies the relations of scale and curvature. As a result, we were able to associate the open submanifold of the curvature space where $\kappa\neq 0$ with the similitude group and provide a left-invariant basis of the Engel structure. We use this local basis to characterize the receptive profiles.
Finally, the behaviour of  global generators of $\mathcal{D}$ which form an infinite dimensional algebra will be studied in the future.

\begin{credits}
\subsubsection{\ackname} \footnotesize I would like to thank Professors Citti and Sarti for their valuable insights and helpful discussions throughout the development of this work. 
Supported by MNESYS project PE12 J33C2002970002,
and Regularity problems in sub-Riemannian structures PRIN 2022 F4F2LH - CUP J53D23003760006

\end{credits}


\begin{thebibliography}{8}
\bibitem{Uncertainty}Barbieri, D., Citti, G., Sanguinetti, G., Sarti, A.: An uncertainty principle underlying the functional architecture of V1. J. of Physiol.-Paris, \textbf{106}(5-6), 183-193, (2012).

\bibitem{Daug}Daugman, J. G.: Uncertainty relation for resolution in space, spatial frequency, and orientation optimized by two-dimensional visual cortical filters. J. of the Optical Society of America, Optics and image science \textbf{2} 7, 1160-9, (1985).



\bibitem{Liouville}Galyaev I, Mashtakov A.: Liouville Integrability in a Four-Dimensional Model of the Visual Cortex. J Imaging. \textbf{7}(12), 277 (2021).

\bibitem{LogPolar}Chessa M, Maiello G, Bex PJ, Solari F.: A space-variant model for motion interpretation across the visual field. J Vis., \textbf{16}(2), (2016).

\bibitem{CortArch}Citti, G., Sarti, A.: Cortical Functional architectures as contact and sub-Riemannian geometry. Morphology, Neurogeometry, Semiotics: A Festschrift in Honor of Jean Petitot's 80th Birthday, Cham: Springer Nature, 111-131, Switzerland (2024).

\bibitem{FuncArch}Citti, G., Sarti, A.: A Cortical Based Model of Perceptual Completion in the Roto-Translation Space. J Math Imaging Vis. 24, 307–326 (2006).

\bibitem{NeuroBook}Citti, G., Sarti, A. (eds.): Neuromathematics of vision. Vol. 32. Springer, Berlin (2014).

\bibitem{Hubel}Hubel, D.H. and Wiesel, T.N.: Ferrier lecture-Functional architecture of macaque monkey visual cortex. Proceedings of the Royal Society of London. Series B. Biological Sciences, \textbf{198}(1130), 1-59. (1977)

\bibitem{LioMarc}Liontou, V., Marcolli, M.: Gabor frames from contact geometry in models of the primary visual cortex. Mathematical Neuroscience and Applications 3, (2023).

\bibitem{Montgomery}Montgomery, R.: Engel deformations and contact structures. Translations of the American Mathematical Society-Series \textbf{2}(196), 103-118, (1999).

\bibitem{Pet}Petitot, J.: Neurogéométrie de la vision, Modèles mathématiques et physiques des architectures fonctionnelles. Éditions de l’École Polytechnique, Palaiseau, (2008)

\bibitem{Tondut}Petitot, J., Tondut, Y.: Geometrie de contact et Champ d’association dans le cortex visuel 9725, CREA, Ecole Polytechnique, Paris, (1997).

\bibitem{Scale2}Sagiv, C., Sochen, N.A., Zeevi, Y.Y.: Scale-Space Generation via Uncertainty Principles. In: Kimmel, R., Sochen, N.A., Weickert, J. (eds) Scale Space and PDE Methods in Computer Vision. Scale-Space 2005. Lecture Notes in Computer Science, 3459. Springer, Berlin, Heidelberg, (2005).

\bibitem{Scale}Sarti, A., Citti, G., Petitot, J.: The Symplectic Structure of the Primary Visual Cortex. Biol Cyb. 98, 33-48, (2008).
\bibitem{Duits}Upanshu, S.,  Duits R.: Left-invariant evolutions of wavelet transforms on the similitude group. Appl. Comput. Harmonic Anal. \textbf{39}(1), 110-137, (2015).


\end{thebibliography}
\end{document}